\documentclass[twocolumn]{revtex4-2}
\usepackage{amsfonts,amssymb,amsmath}            
\usepackage{lmodern}
\usepackage[]{graphics,graphicx}            
\usepackage{amsthm}
\usepackage{bbold,enumitem,mathtools}
\usepackage{tikz}
\usetikzlibrary{decorations.pathmorphing,patterns,decorations.markings,matrix,quantikz}

\usepackage{url}
\usepackage{natbib}
\usepackage{hyperref}
\usepackage[english]{babel}
\makeatletter
\def\bbl@set@language#1{%
  \edef\languagename{%
    \ifnum\escapechar=\expandafter`\string#1\@empty
    \else\string#1\@empty\fi}%
  \@ifundefined{babel@language@alias@\languagename}{}{%
    \edef\languagename{\@nameuse{babel@language@alias@\languagename}}%
  }%
  \select@language{\languagename}%
  \expandafter\ifx\csname date\languagename\endcsname\relax\else
    \if@filesw
      \protected@write\@auxout{}{\string\select@language{\languagename}}%
      \bbl@for\bbl@tempa\BabelContentsFiles{%
        \addtocontents{\bbl@tempa}{\xstring\select@language{\languagename}}}%
      \bbl@usehooks{write}{}%
    \fi
  \fi}
\newcommand{\DeclareLanguageAlias}[2]{%
  \global\@namedef{babel@language@alias@#1}{#2}%
}
\makeatother

\DeclareLanguageAlias{en}{english}

\newtheorem{theorem}{Theorem}

\newtheorem{definition}{Definition}
\newtheorem{lemma}{Lemma}

\newtheorem{assumption}{Assumption}

\newcommand{\identity}{\mathbb{1}}
\renewcommand{\epsilon}{\varepsilon}

\begin{document}

\title{The Smallest Code with Transversal \texorpdfstring{$T$}{T}}
\date{\today}
\author{Stergios \surname{Koutsioumpas}}
\author{Darren \surname{Banfield}}
\author{Alastair \surname{Kay}}
\affiliation{Royal Holloway University of London, Egham, Surrey, TW20 0EX, UK}
\email{Stergios.Koutsioumpas.2018@live.rhul.ac.uk}
\begin{abstract}
We prove that the smallest distance $3$ Quantum Error Correcting Code with a transversal gate outside the Clifford group is the well-known 15-qubit Reed-Muller code, also known as a tri-orthogonal code. Our result relies on fewer assumptions than previous works. We further extend this result by finding the minimal code that also possesses any other non Clifford transversal single-qubit gate.
\end{abstract}
\maketitle

Quantum Computation is receiving a huge focus of international research efforts at the moment, with practical devices being on the verge of achieving ``quantum supremacy'', in which a calculation is performed which has no reasonable hope of being implemented using existing classical algorithms on even the most powerful supercomputers \cite{preskill2012b, arute2019a, madsen2022a}. There have been many suggested quantum algorithms which would provide a substantial speed up compared to known classical ones already \cite{shor1997,bernstein1997}.

The main problems current approaches are facing are the noise from their environment and limited physical resources. Every application of a quantum gate can introduce errors which may cause the computation to fail \cite{aharonov1996}. The aim for the future is sufficient physical resources to implement Fault Tolerant Quantum Computing, which carefully choreographs gate sequences ensuring that any error does not propagate catastrophically and eventually results in Quantum Error Correction \cite{calderbank1996}. 

Thanks to the Quantum Threshold Theorem \cite{aliferis2006}, we know that we can achieve arbitrary accuracy in any calculation provided that the error rate of the components of the computer is below a certain threshold. Hence, the scope for near-term quantum computation on a large scale is heavily dependant upon both the fault tolerant threshold $\varepsilon$ and the resources required to realise the error correcting scheme. A worldwide goal, therefore, is to make the threshold as large as possible. The first rigorous estimates of $\varepsilon=2.73\times10^{-5}$ were made in \cite{aliferis2006}, and have been revised many times, with current state of the art being about $\varepsilon=6.0\times10^{-3}$ \cite{fowler2009}. Beyond these results, however, one can also consider the scaling of resources required to operate close to the threshold obtaining higher estimates such as \cite{knill2005}, as well as any corresponding trade-offs since different error correcting schemes, with new gate sets, and hence different synthesis options \cite{kliuchnikov2015,kliuchnikov2013}, become possible.

There are certain insights that we can derive from \cite{aliferis2006} about how to optimise the threshold:
\begin{itemize}
\item Use as few qubits as possible.
\item Use as simple a circuit for each logical gate as possible with few interactions between qubits.
\end{itemize}
It is also helpful to include more operations than the minimum required for a universal gate set in order to obtain a denser subset of all gates and hence shorter approximation sequences.
 
There is tension between these different constraints. For example, the perfect quantum code \cite{laflamme1996a} uses only five physical qubits per logical qubit, but the gate sequences for a universal set are complicated and have a lot of interactions between qubits, leading to costly propagation of errors. Conventional wisdom has converged on the solution of employing the Steane $[[7,1,3]]$ code. This is the smallest CSS code, and all but one of its gates can be implemented in the simplest way possible,  transversally. Moreover, the Eastin-Knill theorem \cite{eastin2009} tells us that it is impossible to implement all gates from a universal set transversally, hence the Steane code has the minimum amount of non-transversal gates possible for a universal gate set encoding. Another result from the Eastin-Knill theorem is that any error correcting code will only be able to transversally implement a finite set of gates, which are traditionally considered to be the Clifford gates $H, c-\textsc{not}, Z, X, S$. We then require a gate outside this set to achieve universality \cite{barenco1995a}. Potential major contributors to the fault-tolerant threshold are then the implementation of the one non-transversal gate, the $\pi/8$ phase or $T$ gate, which is achieved via magic state distillation (and hence is also very resource intensive) or the c-\textsc{not} gate by nature of being a $2$ qubit gate.

We intend to initiate a review of the alternatives, reconsidering the trade-offs between the listed insights. In particular, the $T$ gate implementation is extremely costly. That said, it could be worse. The gate is at the third order $C_3$ of the Gottesman-Chuang hierarchy \cite{gottesman1999} which is the minimum required as $C_2$ are the Clifford gates. Gates higher in the hierarchy would be even more difficult to implement due to the requirement of extra ancillas. It does, however, open the question of whether we could find a code where the only non-transversal gate is a single-qubit gate at the second order of the Gottesman-Chuang hierarchy. Is any corresponding expansion in the number of qubits in the code compensated for by ease of implementation of the gates, and perhaps having more gates from which to synthesise circuits? One such example was discussed in \cite{knill2005}, but is this the smallest code that could be used, or are there smaller ones, which would lead to a higher threshold?

In this paper, we will prove that the smallest distance $d\geq 3$ error correcting code that possesses transversal $T$ is the Reed-Muller $[[15,1,3]]$ code. Such proofs have previously appeared in the literature \cite{haah2018}\cite{rengaswamy2020}, but have made a variety of restrictive assumptions that we will eliminate, including restricting to non-degenerate CSS codes, and making a stronger assumption about the nature of transversal gates. Knowing the smallest possible code puts one in a good position to proceed with the rest of the review of the fault-tolerant architecture.

Section \ref{sec:notation} begins by providing an overview of the notation and terms used throughout the paper. In section \ref{sec:korthog} we introduce a particular class of error correcting codes, the $k$-orthogonal codes and in \ref{subsec:smallestkorthog} we prove a lower bound for their size. In section \ref{subec:subdual}, we show that the $[2^{k+1}-1,2^{k+1}-2-k,3]$ classical Hamming codes are $k$-orthogonal and achieve the lower bound outlined in \ref{subsec:smallestkorthog}. We then provide a construction of CSS codes based on the Hamming codes which we refer to as sub-dual. We show that these codes transversally implement $P(2\pi/2^{k-1})$ gates. In section \ref{sec:assumptions}, we clearly list and provide the motivation behind our assumptions for the final part of the proof. Finally, in section \ref{sec:smallestphase} we use these assumptions to prove that the minimum size of codes with transversal $P(2\pi/2^{k-1})$ gates is the same as the lower bound of triorthogonal codes and is achieved by the Hamming sub-dual code construction for each $k$.

\section{Notation and definitions}\label{sec:notation} 
In this paper, we will primarily be interested in the action of phase gates on logical qubits. Since those phase gates are diagonal with respect to some computational basis, we focus on their actions on basis states. We have introduced some notation that befits manipulating these basis states, as well as a standard form for codes to make later definitions clear.   
We define the dot product between two binary strings $x,z \in \{0,1\}^n$ as: 
$$
x\cdot z:=(x_1z_1,x_2z_2,\ldots,x_nz_n),
$$
which is equivalent to identifying the bits where both $x$ and $z$ are 1, and
$$
|x|=\sum_{i=1}^nx_i
$$
is the weight of the bit string.

We will use the notation
$$
P({\theta})=\left(\begin{array}{cc} 1 & 0 \\ 0 & e^{i\theta} \end{array}\right)
$$
to define a phase gate. Hence, $P({\theta})^{\otimes n}\ket{x}=e^{i\theta |x|}\ket{x}$

\begin{definition}[Standard Form]
	For a stabilizer code, we can express the stabilizers (up to some possible $\pm i$ phases) as tensor products of $X$ and $Z$ operators, and thus these can be represented by a generator matrix $A$ of $2n$ columns, the first $n$ describing locations of $X$ operations, and the second $n$ describing locations of $Z$ operations. We define a standard form
	$$
	A=\begin{bmatrix}
		A_X & B \\
		0 & A_Z
	\end{bmatrix},
	$$
	where we have redefined the generators via linear combinations in order to reduce the rows that contain $X$s to the minimum possible, i.e. the matrix $A_X$, of $m$ rows, is full rank (modulo 2).
	
	In standard form, $Z_L=Z_r$, $X_L=X_s$ for some binary strings $r$ and $s$ and all the stabilizers corresponding to rows of $A_z$ are all of the form $Z_z$ (as compared to $-Z_z$) for $z\in A_z$.
\end{definition}

In the case of CSS codes, $B=0$. 

In \cite{zeng2011,anderson2016}, it is stated that all single-qubit transversal operations are of the form
$$
L\left(\bigotimes_{i=1}^nP(\theta_i)\right)R^\dagger,
$$ 
where $L$ and $R$ are qubit-wise Clifford operators. There is a lot of freedom to define new codes with identical properties. In particular, if there is a code with the above logical gate, there is an equivalent code with the logical gate:
$$
\bigotimes_{i=1}^nP(\theta_i).
$$
We choose to call this gate a logical phase gate (which automatically implies that $Z_L=Z_r$). There remains further freedom to redefine the code within the $X/Y$ plane. This means that (i) we can choose an equivalent code with $X_L=X_s$, and (ii) we can ensure that all the $Z$-type stabilizers specified by $A_z$ are all of the form $Z_z$ (as compared to $-Z_z$). This second item is assured by finding an $X_y$ which anti-commutes with all the terms $-Z_z$ and commutes with all others. If we pre- and post-multiply by this term, the signs of the stabilizers are correctly updated. In terms of the logical operator, it would only change $\theta_i\mapsto-\theta_i$ on some sites, and hence the general form we are using is not affected.

\subsection{Transversal Gates}\label{subsec:transgates}
We will now carefully look into the definitions of Transversality. When we define an error correcting code, we must also specify how to fault tolerantly implement a set of gates on the logical qubits. Fault tolerance is the property that any computation can be made arbitrarily accurate, given that the physical error rates are below a certain threshold. For distance $3$ error correcting codes, fault tolerance is guaranteed by circuits where any fault in a physical qubit only leads to at most one error in each encoded block of qubits \cite{gottesman1997}. The simplest constructions are the transversal implementations. 

\begin{definition}
A logical single qubit unitary is implemented in a \emph{transversal} manner if it is implemented by individual operations on each qubit $i$. For multi-qubit logical gates, the transversal implementation means that unitaries must be applied only on sets of qubits (from different logical qubits) with the same label, e.g. an operation on all the physical qubits labelled `1'.
\end{definition}


\begin{definition}
We say that a gate is \emph{strongly transversal} if the operation on each set of identically labelled qubits is the same for each and every label.
\end{definition}
For example, $S\otimes I \otimes T$ is a transversal but not strongly transversal operation while $T\otimes T\otimes T$ is strongly transversal on a code of $3$ qubits.

The strongly transversal assumption is often made (some papers, such as \cite{rengaswamy2020}, partially remove the assumption in the specific case of transversal $T$). We will avoid this assumption.

Note that we will assume that the logical $X$ and $Z$ operators are transversal, but not that they are strongly transversal. We will denote them by $X_L=X_s$, where $s\in\{0,1\}^n$ is the binary string which has $0$ in every position where the identity is applied on a physical qubit and $1$ where $X$ is applied.  

With these definitions in mind, we can look into a particular class of classical error correcting codes.

\section{\texorpdfstring{$k$}{k}-orthogonal Codes}\label{sec:korthog}
In this section we introduce the classical $k$-orthogonal codes and find some bounds on their size. In later sections, we show that these codes can be used to define quantum codes and are tied to the transversal implementation of any $P(2\pi/2^{k-1})$ phase gate.
\begin{definition}
	We say that a code is $k$-orthogonal if:
	\begin{align*}
		|x^1|&\equiv 0\text{ mod }2 \\
		|x^1\cdot x^2|&\equiv 0\text{ mod }2 \\
		|x^1\cdot x^2\cdot x^3|&\equiv 0\text{ mod }2 \\
		&\vdots \\
		|x^1\cdot x^2\cdot x^3\cdot\ldots\cdot x^k|&\equiv 0\text{ mod }2 \\
	\end{align*}
	for all $x^i\in S_X$, which is the group of strings generated by $A_X$ under addition modulo 2. Unless otherwise stated, the code is implied to be distance $d\geq 3$ and non-degenerate.
\end{definition}

\begin{lemma}\label{lem:korthogonal}
All $k$-orthogonal codes (that are non-degenerate and have distance $3$) have a parity check matrix $A_X$ with $m$ rows such that $m>k$.
\end{lemma}
\begin{proof}
Let $N_q$ be the number of columns of $A_X$ for which the first $q$ rows contain at least one non-zero entry. So, $N_1=|x^1|$. Evaluating $N_2$, we can find $|x^1|$ and $|x^2|$ and subtract the number of sites that are double counted: $|x^1\cdot x^2|$. Hence,
$$
N_2=|x^1|+|x^2|-|x^1\cdot x^2|.
$$
Similarly,
$$
N_3=|x^1|+|x^2|+|x^3|-|x^1\cdot x^2|-|x^1\cdot x^3|-|x^2\cdot x^3|+|x^1\cdot x^2\cdot x^3|
$$
In general, $N_i$ depends on a single $i$-product and other lower order products.

Now select any column $q$. We can update our choices of rows in the generator $A_X$ by taking linear combinations of existing rows. In particular, pick any row $y\in A_X$ for which $y_q=1$. Since the code is distance 3, we are guaranteed that at least one row satisfies this (because it yields a syndrome when a $Z_q$ error occurs). Now replace all $x\in A_X:x_q=0$ with $y\oplus x$. This means that when we write out the binary matrix $A_X$, the column $q$ is the all ones column. Since the code is non-degenerate, it must be the \emph{only} column which is all-ones. Hence $x_1\cdot x_2\cdot x_3\ldots x_m=000\ldots 0100\ldots 0$, where the dot product is taken over all rows and the 1 is in the position $q$. Consequently, $|x_1\cdot x_2\cdot x_3\ldots x_m|=1$, and so the code is not $m$-orthogonal.
\end{proof}

\subsection{The smallest \texorpdfstring{$k$}{k}-orthogonal Code}\label{subsec:smallestkorthog}
We will now prove the minimal size relation for these codes. Later, in theorem \ref{thm:korthog} we will show that $k$-orthogonality is a necessary condition for certain transversal phase gate implementations, thus giving some upper bounds on the minimal size of a code transverally implementing these gates. 
 
\begin{theorem}\label{lem:smallestkorth}
The smallest $k$-orthogonal code contains at least $2^{k+1}-1$ bits.
\end{theorem}
\begin{proof}
Assume our code is $p$-orthogonal ($k\leq p< m$) but not $(p+1)$ orthogonal. This means that there is a choice of $(p+1)$ rows such that $|x^1\cdot x^2\cdot\ldots\cdot x^{p+1}|$ is odd. By Lemma \ref{lem:korthogonal}, such a $p$ exists. We choose to arrange $A_X$ so that these $(p+1)$ rows are the first $(p+1)$ rows of $A_X$. Thus, $N_1\equiv N_2\equiv N_3\equiv\ldots\equiv N_p\equiv 0\text{ mod }2$ and $N_{p+1}\equiv 1\text{ mod }2$.

We now try to construct a specific instance of these $(p+1)$ rows of the code. An odd number of columns are the all-ones (identified by $x_1\cdot x_2\cdot x_3\ldots x_{p+1}$). Since we want the minimal size case, let us take just a single column. Next, consider removing any one row. Applying $p$-orthogonality, there is at least one more column that is all-ones in the remaining columns. But we can repeat this for every row that we remove. Hence, we have a structure that looks like
$$
\begin{array}{c|cccc}
1 & 1 & 1 & 1 & 0 \\
1 & 1 & 1 & 0 & 1 \\
1 & 1 & 0 & 1 & 1 \\
1 & 0 & 1 & 1 & 1
\end{array}
$$
This example has all triple products with even weight, but has not yet considered the double products. These, currently, are all odd-valued. So, we'll have to add extra columns. To make those work, we end up adding
$$
\begin{array}{c|cccc|cccccc}
1 & 1 & 1 & 1 & 0 & 0 & 0 & 0 & 1 & 1 & 1 \\
1 & 1 & 1 & 0 & 1 & 0 & 1 & 1 & 0 & 0 & 1\\
1 & 1 & 0 & 1 & 1 & 1 & 0 & 1 & 0 & 1 & 0\\
1 & 0 & 1 & 1 & 1 & 1 & 1 & 0 & 1 & 0 & 0
\end{array}.
$$
But this still leaves the weight of individual rows odd. We have to add yet more columns
$$
\begin{array}{c|cccc|cccccc|cccc}
1 & 1 & 1 & 1 & 0 & 0 & 0 & 0 & 1 & 1 & 1 & 1 & 0 & 0 & 0 \\
1 & 1 & 1 & 0 & 1 & 0 & 1 & 1 & 0 & 0 & 1 & 0 & 1 & 0 & 0\\
1 & 1 & 0 & 1 & 1 & 1 & 0 & 1 & 0 & 1 & 0 & 0 & 0 & 1 & 0\\
1 & 0 & 1 & 1 & 1 & 1 & 1 & 0 & 1 & 0 & 0 & 0 & 0 & 0 & 1
\end{array}.
$$
In all, we end up with the minimal size case with 
$$n={p \choose p} + {p \choose p-1} + \cdots + {p \choose 2} + {p \choose 1} = 2^{p+1}-1$$ So, for a $k$-orthogonal code we require $n\geq2^{k+1}-1$. This construction gives us exactly the limiting cases of these codes, and they coincide with the parity check matrices of the Hamming $[2^{k+1}-1,2^{k+1}-2-k,3]$ codes.
\end{proof}

\subsection{The sub-dual Hamming codes}\label{subec:subdual}
We will ultimately be considering transversal gates on Quantum Error Correcting Codes based on the $k$-orthogonal codes. 

\begin{lemma}
	There exists a $2^{k+1}-1$ qubit CSS code which is $k$-orthogonal.
\end{lemma}
The nice thing about having CSS codes is that it guarantees we have transversal controlled-not as well, maximising the sets of gates we can use.

For odd $n$, if we take an $A_X$ that is $k$-orthogonal ($k\geq 2$), we can construct a CSS quantum code for it by fixing $s=111\dots 1=r$ and finding the $A_Z$ that is the null space of $A_X$ and $r$ (this is just the way that tri-orthogonal codes have previously been constructed \cite{nezami2022}). Then, due to bi-orthogonality, $A_X$ is in that null space and hence $A_X\subseteq A_Z$.

Take the $[2^{k+1}-1,2^{k+1}-2-k,3]$ classical Hamming code. Use the parity check of this as $A_X$ in a CSS code. Impose that the quantum code will encode a single qubit and have $X_L=X^{\otimes n}$ and $Z_L=Z^{\otimes n}$. This basically fixes $A_Z$, using the null space of $A_X$ combined with the logical operator.

In appendix \ref{appendixA} we will show that while the code is of distance $3$ with respect to $Z$ errors, it is of distance $2^{k}-1$ with respect to $X$ errors. This additional distance can prove beneficial in some scenarios. 

\section{Initial Assumptions}\label{sec:assumptions}

There are several assumptions that we will make that strongly influence the properties of the code that we will choose. We list them here carefully for clarity.

\begin{assumption}
The set of logical operators with transversal implementation will include a phase gate $P(\theta)$ other than $S,Z,S^\dagger$.
\end{assumption}
By construction, our code in standard form will have a transversal phase gate. We need one gate that is outside of the Clifford group, and it is this one that we are choosing to select (any other choice would be represented in this form via conversion to standard form).
\begin{assumption}\label{assumption}
The transversal implementation of that gate must only comprise physical application of that gate (and its powers).
\end{assumption}
Consider a transversal gate at the top level of concatenation. We require the transversality to propagate down through all layers of the hierarchy. This means that every gate we use must have a transversal implementation at all levels. For example if our code is equipped with transversal $T$ but we need Steane's \cite{steane1997a} or Shor's \cite{shor1997a} error correction procedure to obtain $H$, then although we would regard a logical gate which could be implemented at level-$1$ of our code as $T\otimes H\otimes T$ as transversal, at the second level of concatenation its implementation requires a level-$1$ $H$
gate which would not be transversal. This assumption ensures transversality on every level of concatenation.
\begin{assumption}
The codes that we consider will contain a single logical qubit.
\end{assumption}
This is in no way limiting --- if there is a code that encodes more than one qubit, one of which has transversal $T$, we just consider a new code that is the same as the original, but where the logical operators of all the other qubits now become stabilizers. The code distance cannot decrease. This code has a single logical qubit with transversal $T$.

\section{The smallest codes with transversal $\frac{\pi}{2^{k-1}}$ phase gates}\label{sec:smallestphase}
We are now ready to link the previous results. So far we have proved that the minimum size of a classical $k$-orthogonal code is $2^{k+1}-1$ bits, and that for each $k$ there exists a CSS $k$-orthogonal code which is $2^{k+1}-1$ qubits long. Below we will start by looking into how a general transversal phase gate will look. We then use this to prove that $k$-orthogonality on a subset of the code is a necessary condition for a transversal phase gate, thus concluding that the minimum size for a code with transversal  $\frac{\pi}{2^{k-1}}$ phase is $2^{k+1}-1$ qubits as well. We then generalise the result to account for controlled-phase gates, as well.

\begin{theorem}
A transversal phase gate $P(\theta)$ for any non-degenerate stabilizer code of distance $d\geq 3$ must have a phase $\theta=p\frac{\pi}{2^{k-1}}$ for integers $p,k$ and $k\leq m-1$.
\end{theorem}
Already proven in \cite{anderson2016}, we will give a proof based on our notation and for our case only.
\begin{proof}
We construct the projector onto logical 0, and project onto the all-zeros state\footnote{This only works because of our chosen standard form.}. Up to normalisation, this defines
$$
\ket{0_L}=(\identity+Z_L)\prod_{(x,y)\in A_XB}(\identity+s_{xy}X_xZ_y)\prod_{z\in A_z}(\identity+Z_z)\ket{0}^{\otimes n},
$$
where $s_{xy}$ is a possible value $\pm1,\pm i$, which is not important for our purposes. The $Z$-only terms immediately vanish, leaving
$$
\ket{0_L}=\prod_{(x,y)\in A_XB}(\identity+s_{xy}X_xZ_y)\ket{0}^{\otimes n}=\frac{1}{\sqrt{2^k}}\sum_{x\in S_x}\tilde s_x\ket{x},
$$
where $S_x$ is the group of strings generated by $A_X$ and $\tilde s_x$ is an updated phase value. The important feature of our choice of $A_X$ being full rank is that it assures that there are no products of terms $X_xZ_y$ that all yield the same $X$ components, which would potentially annihilate different terms.

Let us apply phase gates $P(\theta_i)$ to each qubit. If this is to achieve logical $T$, we require that
$$
\left(\bigotimes_iP(\theta_i)\right)\ket{x}=e^{i\gamma}\ket{x}
$$
for all $x\in S_x$. Since $000\ldots 0\in S_x$, it must be that $\gamma=0$. Hence, we require that
$$
\sum_{i=1}^nx_i\theta_i=|x\cdot\vec{\theta}|\equiv 0\text{ mod }2\pi
$$
for all $x_i\in S_X$.

Now, consider any two $x^1,x^2\in S_x$. It is also the case that $x^1\oplus x^2\in S_x$. However, we can rewrite
$$
|(x^1\oplus x^2)\cdot\vec{\theta}|=|x^1\cdot\vec{\theta}|+|x^2\cdot\vec{\theta}|-2|(x^1\cdot x^2)\cdot\vec{\theta}|.
$$
This means that
$$
2|(x^1\cdot x^2)\cdot\vec{\theta}|\equiv 0\text{ mod }2\pi.
$$
By induction, we have
$$
2^{m-1}|(x^1\cdot x^2\cdot\ldots\cdot x^m)\cdot\vec{\theta}|\equiv 0\text{ mod }2\pi.
$$

Now, however, let us use the fact that the code is distance 3 and non-degenerate. By the construction of Lemma \ref{lem:korthogonal}, select a single site $i$. The error $Z_i$ must have a unique non-trivial syndrome. We can convert this into a set of $m$ rows (replacing the previous definition of $A_X$) such that every element in column $i$ is 1. Since the code is non-degenerate, it must be the only column for which this is true. Hence, if we use these rows to calculate $x^1\cdot x^2\cdot\ldots\cdot x^m$, we isolate just the column $i$. Thus,
$$
2^{m-1}\theta_i\equiv 0\text{ mod }2\pi
$$
for all $i$. This fixes the form of the phase to be $\frac{p_i\pi}{2^{m-2}}$.

Moreover, $\ket{1_L}=X_s\ket{0_L}$. This means that the net phase which is applied to this logical state is
$$
\theta=\frac{\pi}{2^{m-2}}\sum_{i=1}^np_is_i,
$$
limiting the form of the logical $P(\theta)$ that can be implemented.
\end{proof}

\emph{Technically}, this result could, for example, allow us to create a logical $\pi/4$ gate by transversal application of $\pi/8$ gates for any code where $m>k+1$. While it seems likely that the minimal case will keep $m$ as small as possible, and hence $m=k+1$, we nevertheless explicitly impose via assumption \ref{assumption} that this does not happen since it would not permit transversal implementation throughout the concatenated hierarchy.

\begin{lemma}
For any degenerate code, a subset of qubits can be chosen from whose perspective the code is a non-degenerate code of the same distance, and hence the same conclusions hold.
\end{lemma}
\begin{proof}
A degenerate code is characterised by the fact that several single-qubit errors may have the same syndrome. We can divide qubits up into degeneracy sets $\Lambda_i$, meaning that each $\Lambda_i$ is the set of qubits $j$ for which the error $Z_j$ has the same syndrome. These disjoint sets are not affected by redefinitions of $A_X$, and it means that for any product $x_1\cdot x_2\cdot \ldots x_k$, the bit values $(x_1\cdot x_2\cdot \ldots x_k)_j$ are equal for all $j\in\Lambda_i$. As such, any value
$$
\sum_{i=1}^nx_i\theta_i
$$
is characterised only by the values $\Theta_i=\sum_{j\in\Lambda_i}\theta_j$. Hence, let us select $v\in\{0,1\}^n$ such that $v_j=1$ for exactly one member $j$ of each $\Lambda_i$. Without loss of generality, we can choose to apply the phase $\Theta_i$ to the corresponding qubit indicated by $v$, and 0 on all other members of $\Lambda_i$.

Now, the logical operation is $Z_r$ with $r\subseteq v$. As the columns $j$ of $A_X$ are all identical for all $j\in\Lambda_i$, but distinct for columns $j\in\Lambda_i,k\in\Lambda_{i'}$ which are members of different sets $i\neq i'$, the columns $r$ of $A_X$ are all distinct. Our previous proofs hold with respect to these columns; it acts like a non-degenerate code.
\end{proof}

\begin{theorem}\label{thm:korthog}
Any code that has transversal $P(\theta)$ where $\theta=\frac{p\pi}{2^{k-1}}$ ($p$ odd) must be $k$-orthogonal with respect to a subset of vertices $r$.
\end{theorem}
Tri-orthogonality was observed in \cite{rengaswamy2020} for $T$ under stronger assumptions (strongly transversal assumed). Here, we obtain a stronger result by weakening the conditions, which coincides when $m=4$.
\begin{proof}
By Assumption \ref{assumption}, we take
$$
\theta_i=\frac{p_i\pi}{2^{k-1}}
$$
Hence, we know that
$$
\sum_i(x_1\cdot x_2\cdot \ldots \cdot x_q)_ip_i\equiv 0\text{ mod }2^{k-q+1}
$$
for all $q=1$ to $k$. Thus, it is also true that
$$
\sum_i(x_1\cdot x_2\cdot \ldots \cdot x_q)_ip_i\equiv 0\text{ mod }2
$$
provided $q\leq k$.

Next, let us observe that if we repeat our gate $2^{k-1}$ times, we must implement logical $Z$. In other words, $p_i\equiv r_i\text{ mod }2$\footnote{Technically, the set of qubits acted upon might not be exactly the set identified by $Z_r$, but this is only because we can redefine $Z_L\mapsto Z_LZ_z$ for any stabilizer $Z_z$.} Hence, we see that $|x_1\cdot x_2\cdot \ldots \cdot x_{q}\cdot r|\equiv 0\text{ mod }2$ for all possible $q$-wise products of strings in $S_X$ for $q=1$ to $k$: $k$-orthogonality with respect to the subset r.
\end{proof}

In order to have a code with transversal $T$, we set $k=3$. The code must be tri-orthogonal. The smallest tri-orthogonal distance 3 code has 15 qubits, and is the well known code from \cite{bravyi2012}.

The following lemma provides a more general result.

\begin{lemma}
A code in standard form has transversal controlled-phase gates with $q$ control qubits and a target phase $P(p\pi/2^{k-q-1})$ only if the code is $k$-orthogonal.
\end{lemma}
This proves that every instance that has a transversal $T$ (for example) also has gates such as controlled-$S$ which would be helpful with gate synthesis. It also has controlled-controlled-phase, but we probably wouldn't include that due to its adverse effect on the fault-tolerant threshold. In turn, this deals with the question of what our two-qubit gate should be --- we are automatically provided with at least one with a transversal implementation (c-\textsc{not}) if we are dealing with CSS codes. Moreover, it proves that we're not missing a trick in the sense of asking for a hypothetical shorter code without transversal $T$, but with transversal controlled-$S$.
\begin{proof}
We assume a structure for the gate of
$$
U=\bigotimes_{i=1}^nc^q-P\left(\frac{p_i\pi}{2^{k-q-1}}\right).
$$
Consider the action on a state $\ket{0_L}^{\otimes(q+1)}$. Since $U$ is diagonal, it maps basis states to basis states. Thus, we require
\begin{equation}\label{eqn:leapfrog}
U\ket{x^1}\ket{x^2}\ldots\ket{x^{q+1}}=\ket{x^1}\ket{x^2}\ldots\ket{x^{q+1}}.
\end{equation}
In other words,
$$
2\pi\sum_i(x^1\cdot\ldots\cdot x^{q+1})_i\frac{p_i}{2^{k-q}}\equiv0\text{ mod }2\pi.
$$
Hence,
$$
\sum_i(x^1\cdot\ldots\cdot x^{q+1})_ip_i\equiv0\text{ mod }2^{k-q}.
$$
We can use our previous tricks to thus argue that
$$
\sum_i(x^1\cdot\ldots\cdot x^{t})_ip_i\equiv0\text{ mod }2^{k+1-t}.
$$
for $t=q+1$ to $k$.

Moreover, there is no need to keep the $x^i$ in Eq.\ (\ref{eqn:leapfrog}) distinct: if we set some set of them all equal, of course $x^1\cdot x^1\cdot x^1=x^1$, and we thus get equivalent statements for shorter sequences:
$$
\sum_i(x^1\cdot\ldots\cdot x^{t})_ip_i\equiv0\text{ mod }2^{k-q}.
$$
for all $t\leq q$.

This can be summarised as
$$
\sum_i(x^1\cdot\ldots\cdot x^{t})_ip_i\equiv0\text{ mod }2^{k-\max(q,t-1)}.
$$
Again, let $r$ contain the information on the parities of $p_i$ (it is not so immediate that this relates to $Z_L=Z_r$). Thus,
$$
|x^1\cdot\ldots x^t\cdot r|\equiv 0\text{ mod }2 \qquad t\leq k.
$$
The code must be $k$-orthogonal with respect to the subset of qubits $r$. Hence $n\geq |r|\geq 2^{k+1}-1$. There are no smaller cases of codes with transversal controlled-phases.
\end{proof}

\section{Conclusion}
We have proven that $k$-orthogonal codes are the minimum codes with a transversal $P({\frac{\pi}{2^{k+1}}})$ gate, thus extending previous results. We started by defining $k$-orthogonality, then looked into CSS codes constructed with $k$ orthogonal codes and showed that they are a subcase of the Hamming code constructions. We made sure to clearly identify the assumptions needed for the last step, where we proved that a transversal $P({\frac{\pi}{2^{k+1}}})$ phase gate induces a $k$-orthogonal structure on at least a subset of the code's vertices, thus showing that the minimum sizes are indeed the same.   

We focused on codes with transversal non-Clifford gates to avoid the commonly used costly distillation techniques, as by using a code with transversal phase gate, we only need to implement a Hadamard gate, whose state can be fault tolerantly prepared \cite{steane2004a, goto2016}. We also saw that the codes with transversal $P({\frac{\pi}{2^{k+1}}})$ phase will also have transversal controlled- $P({\frac{\pi}{2^{k}}})$ phase gates which can be beneficial during the synthesis stage.
 
It would be interesting to see how code distance affects the qubit sizes of our codes, as well as whether that could be beneficial for deriving higher thresholds.

\bibliographystyle{IEEEtran.bst}
\bibliography{MyLibrary1}

\appendix
\section{\texorpdfstring{$X$}{X}-error Distance}\label{appendixA} 
We want to investigate the distance of the sub-dual Hamming codes with respect to $X$ errors. We set
$$
A_X=\left(\begin{matrix} \identity & c & V \end{matrix}\right)
$$
where $c$ is a column of $m$ entries, two of which are non-zero. $V$ is a $n\times (n-m-1)$ matrix whose columns are all of the binary strings of weight 2 or more except for $c$. We then construct
$$
A_Z=\left(\begin{matrix} J & d & \identity \end{matrix}\right)
$$
where $d$ is an $n-m-1$ element column comprising the opposite of the parity of the weights of columns of $V$, and
$$
J=V^T+dc^T\text{ mod }2.
$$
One can readily verify that $A_Z\cdot A_X^T\equiv 0\text{ mod }2$, and that the row weights of $A_Z$ are even, guaranteeing commutation with $X^{\otimes n}$.

By construction, we know that the distance to $Z$ errors is 3: all columns of $A_X$ are distinct, so all single-qubit $Z$ errors have a distinct syndrome, and can hence be corrected, while there are sets of 3 columns that have a trivial sum (e.g.\ $c$ and the corresponding two columns from the $\identity$ term). What about the distance to $X$ errors? We want to find the minimum number columns of $A_Z$ that one can add together before getting a trivial outcome.

Clearly, adding distinct columns just arising from the $\identity$ matrix can never give a trivial output. On the other hand, all other columns have a weight of at least $2^{m-1}-2$: for column $i$ if $c_i=0$, then the weight is the same as the weight of row $i$ of $V$. But the total weight of row $i$ in $A_X$ is $2^{m-1}$, of which 1 is in the $\identity$ term, and, by assumption, the $c_i=0$. (We shall ignore the $c_i=1$ case for now --- it is more complicated without adding anything to the discussion and, after all, this only affects two columns).

As such, one of these columns from $J$ with $2^{m-1}-2$ columns from the identity matrix would give a trivial output. Hence, this suggests the distance could be $2^{m-1}-1$. What about arbitrary combinations of columns from $J$? Again, let us start by ignoring the two columns for which $c_i=1$. We know that the code is $k$-orthogonal ($k=m-1$). This means that all the linear combinations of rows of $A_X$ also have weight $2^{m-1}$. Thus, if we combine $q$ columns of $J$, the weight is $2^{m-1}-q$, and this does not change the count of how many columns are involved in a trivial combination.

Finally, let us consider the effect of those two columns (for which $c_i=1$) and $d$. We can easily verify that the weight of $d$ is the number of columns of $V$ with even weight, which is $2^{m-1}-2$ --- half of all binary strings have an even weight, but we discount $c$ and the all 0s string.

If an even number of these columns is included in our sum, then the occurrences of $d$ in each of them just cancel. In this case, the only difference is that because $c_i$ could be 1, the total weight of the column of $A$ is $2^{m-1}-2$, and the total weight of $q$ such columns combined is no less than $2^{m-1}-1-q$.
\end{document}